\documentclass[11pt]{article}
\usepackage{acl2016}
\usepackage{times}
\usepackage{url}
\usepackage{latexsym}
\usepackage{graphicx}

\usepackage{subfigure}

\usepackage{amsmath, amssymb, times}
\usepackage{amsthm}
\usepackage{algorithm}
\usepackage{algorithmic}
\usepackage{amsmath}
\usepackage{amsfonts}
\usepackage{hyperref}

\aclfinalcopy 

\newcommand{\mathbs}{\boldsymbol}

\newtheorem{lemma}{Lemma}
\newtheorem{theorem}{Theorem}

\title{Solve-Select-Scale: A Three Step Process For Sparse Signal Estimation}

\author{Mithun Das Gupta \\
  {\tt mithundasgupta@gmail.com} }

\date{}

\begin{document}
\maketitle
\begin{abstract}
In the theory of compressed sensing (CS), the sparsity $\|x\|_0$ of the
unknown signal $\mathbf{x} \in \mathcal{R}^n$ is of prime importance and the focus of reconstruction algorithms has mainly been either $\|x\|_0$ or its convex relaxation (via $\|x\|_1$).  
However, it is typically unknown in practice and has remained a challenge when nothing about the size of the support is known. As pointed recently, $\|x\|_0$ might not be the best metric to minimize directly, both due to its inherent complexity as well as its noise performance. 
Recently a novel stable measure of sparsity $s(\mathbf{x}) := \|\mathbf{x}\|_1^2/\|\mathbf{x}\|_2^2$ has been investigated by Lopes~\cite{Lopes2012}, which is a sharp lower bound on $\|\mathbf{x}\|_0$. The estimation procedure for this measure uses only a small number of linear measurements, does not rely on any sparsity assumptions, and requires very little computation. 
The usage of the quantity $s(\mathbf{x})$ in sparse signal estimation problems has not received much importance yet. We develop the idea of incorporating $s(\mathbf{x})$ into the signal estimation framework. We also provide a three step algorithm to solve problems of the form $\mathbf{Ax=b}$ with no additional assumptions on the original signal $\mathbf{x}$.
\end{abstract}

\section{Introduction}\label{SEC:Intro}
Recently many measures of sparsity have been proposed to estimate the support of a vector $\mathbf{x}\in \mathcal{R}^n$ where the support is assumed to be smaller than the data dimension $n$.  
\begin{figure}[ht]
\begin{center}
{\includegraphics[width=8cm]{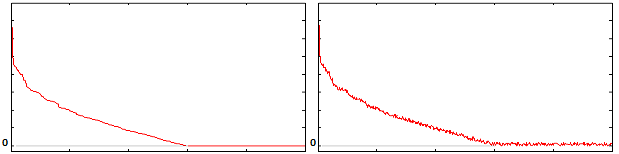}}
\caption{Left: Original sorted signal, and right: noisy version of left. The $\|x\|_0$ norm increases from 300 to 500. The norm ratio $s(\mathbf{x})$ changes from 190 to 210.}
\label{Fig:sx}
\end{center}
\end{figure}
The ratio of the two norms
\begin{equation}\label{Eq:Sx}
 s(\mathbf{x}) := \|\mathbf{x}\|_1^2/\|\mathbf{x}\|_2^2
\end{equation}
has been proven to be useful measure of sparsity by multiple researchers~\cite{Hoyer04,Hurley09,Tang11,Lopes2012}.
Lopes~\cite{Lopes2012} proposed an estimate $\tilde{s}(\mathbf{x})$ for $s(\mathbf{x})$ based on sketching with p-stable distributions. The efficacy of such a measure can be established by looking at the signals presented in Fig.~\ref{Fig:sx}. The $\|x\|_0$ norm of the signal jumps from 300 to $n=500$ (full signal length) with the addition of noise. Intrestingly, sparsity measure $s(\mathbf{x})$ for the noisy signal in Fig.~\ref{Fig:sx} grows from marginally from 190 to 210. The inherent noise robustness for the two metrics can be judged by the fact that the noise level is several dbs lesser than the original signal. We argue that such a sparsity measure needs to be incorporated in sparse signal estimation due to the increased robustness to noise. 

\section{New Direction}
In this paper we would like to estimate the unknown bounded sparse signal $\{\mathbf{x} | \mathbf{x}\in \mathcal{R}^N, \|\mathbf{x}\|_\infty \le M \}$, given linear measurements of the form $\mathbf{Ax=b}$ with \textbf{NO} additional assumptions about the sparsity of the signal $\mathbf{x}$. We assume that the infinity norm bound $M$ has been absorbed into the sensing matrix and henceforth we talk of signals $\|\mathbf{x}\|_\infty \le 1$. We assume $\mathbf{A}\in \mathcal{R}^{M\times N}$ is a Gaussian random measurement matrix.
It has been shown in~\cite{Lopes2012} that $s(\mathbf{x})$ is strictly upper-bounded by the actual $l_0$-norm of a vector, namely
\begin{equation}\label{Eqn:firstIneq}
  s(\mathbf{x}) \le \|\mathbf{x}\|_0
\end{equation}
Since both the quantities in Eq.~\ref{Eqn:firstIneq} are positive quantities with $\|\mathbf{x}\|_0 \ge 1$, taking log on both sides and invoking the Jensen's inequality leads to
\begin{eqnarray}
  && \log( s(\mathbf{x})) \le \log( \|\mathbf{x}\|_0) \\
  &\Rightarrow& \log \|\mathbf{x}\|_1^2 - \log \|\mathbf{x}\|_2^2 \le \log( \|\mathbf{x}\|_0) \\
  \label{Eqn:secondIneq}
  &\Rightarrow& \sum_i^N \log |x_i|^2} - {\log \|\mathbf{x}\|_2^2 \le \log( \|\mathbf{x}\|_0)
\end{eqnarray}
Consequently, we would like to maximize the quantity on the LHS in Eq.~\ref{Eqn:secondIneq} while maintaining the data term $\|\mathbf{Ax - b}\|^2$ as low as possible. The maximization can be turned into an alternate minimization problem with a switch of signs. Before proceeding to the optimization technique, we deal a little longer on the properties of the new term which we propose to minimize.

\begin{theorem}
Let
\begin{equation}\label{EQN:SparseMetric}
f(\mathbf{x}) = -\sum_i^n \log |x_i|^2 + \log \|\mathbf{x}\|_2^2
\end{equation}
then the Hessian $\mathbf{H}$ of $f(\mathbf{x})$ is positive definite for any bounded signal $\{\mathbf{x}\in \mathcal{R}^n|~\|\mathbf{x}\|<\infty, x_i > 0 \forall i \}$.
\end{theorem}

\begin{proof}
Without loss of generality let us assume $\|\mathbf{y}\|_2 = 1$. 
We need to prove that 
\begin{equation}
\mathbf{y}^T\mathbf{Hy} = \sum_i^n y_i^2 h_{ii} + \sum_i^n \sum_j^n y_i y_j h_{ij} > 0
\end{equation}
for any $\mathbf{y}$, where $h_{ii}$ is the $i^{th}$ diagonal element and $h_{ij}$ is the element at position $(i,j)$ of the matrix $\mathbf{H}$. The elements of $\mathbf{H}$ may in turn be written as
\begin{eqnarray}
\label{Eqn:h_ii}
h_{ii} & = &  2\frac{(S_2-x_i^2)(2x_i^2+ S_2)}{x_i^2S_2^2}\\ 
\label{Eqn:h_ij}
h_{ij} & = &  -4\frac{x_ix_j}{S_2^2}\\
S_2 &=& \sum x_i^2
\end{eqnarray}
For the $1^{st}$ row of $\mathbf{H}$ we can write the summation as
\begin{eqnarray}
&y_1^2h_{11} + \sum_j y_1 y_j h_{1j} = &  \\ \nonumber
& \frac{2}{x_1^2S_2^2} 
	\begin{pmatrix}
       y_1^2 (x_2^2+\ldots+x_n^2)(3x_1^2+x_2^2+\ldots+x_n^2)  \\
	-2y_1 y_2x_1^3x_2 \\
	-2y_1 y_3x_1^3x_3 \\
	\vdots \\
	-2y_1 y_nx_1^3x_n \\
     \end{pmatrix}&
\end{eqnarray}
which can be re-written as
\begin{eqnarray}
& y_1^2h_{11} + \sum_j y_1 y_j h_{1j} = & \\ \nonumber
& \frac{2}{x_1^2S_2^2} 
	\begin{pmatrix}
       y_1^2(x_2^2+\ldots+x_n^2)^2  \\
	+3y_1^2x_1^2(x_2^2+\ldots+x_n^2)  \\
	-2y_2x_1^2y_1x_1x_2 - \ldots -2y_nx_1^2y_1x_1x_n \\
     \end{pmatrix} &\\ \nonumber
& =  \frac{2}{x_1^2S_2^2} 
	\begin{pmatrix}
       y_1^2(x_2^2+\ldots+x_n^2)^2  \\
	+2y_1^2x_1^2(x_2^2+\ldots+x_n^2)  \\
	-(y_2^2 + y_3^2 + \ldots + y_n^2)x_1^4\\
	+(y_2x_1^2 - y_1x_1x_2)^2 \\
	\vdots \\
	+ (y_nx_1^2 - y_nx_1x_n)^2  \\
     \end{pmatrix} &
\end{eqnarray}
taking the second and third terms from the bracket separately, we can write
\begin{eqnarray}
&y_1^2h_{11} + \sum_j y_1 y_j h_{1j} = & \\ \nonumber
& [\mathrm{+ve ~terms}] + \frac{2}{S_2^2} 
	\begin{pmatrix}
      2y_1^2(x_2^2+\ldots+x_n^2)  \\
	-(y_2^2 + y_3^2 + \ldots + y_n^2)x_1^2 \\
     \end{pmatrix} &\\ \nonumber
& = [\mathrm{+ve ~terms}] + \frac{2}{S_2^2} 
	\begin{pmatrix}
      2y_1^2(x_1^2 + x_2^2+\ldots+x_n^2)  \\
	-(y_1^2+y_2^2  + \ldots + y_n^2)x_1^2 \\
	- y_1^2x_1^2 \\
     \end{pmatrix}&
\end{eqnarray}
which can be succintly written as
\begin{equation}
y_1^2h_{11} + \sum_j y_1 y_j h_{1j} = \frac{4y_1^2}{S_2} -  \frac{2(1+y_1^2)x_1^2}{S_2^2} + [\boldsymbol{+}]
\end{equation}
where $[\boldsymbol{+}]$ denotes positive terms and the simplification arrived at by using $S_2 = \sum x_i^2$ and $\|\mathbf{y}\|_2=1$.
Adding the contributions of all the rows of $\mathbf{H}$ we get
\begin{equation}
\mathbf{y}^T\mathbf{Hy} = \frac{2}{S_2^2}\sum (1-y_i^2)x_i^2 + [\boldsymbol{+}]> 0 
\end{equation}
\end{proof}

\section{Optimization}\label{SEC:Optimization}
We have established that for positive elements of the unknown vector, the cost function allows a strictly positive definite Hessian. For elements of the unknown vector, which are zero or negative, the cost function is not defined. We now propose a variable separation technique, which takes care of both the zero as well as the negative components of the unknown vector in a principled way. The optimization problem till this point can be written as
\begin{eqnarray}
  \min_\mathbf{x} && -\sum_i \log |x_i|^2 + {\log \|\mathbf{x}\|_2^2} \\
  s.t. &&  \mathbf{Ax = b}
\end{eqnarray}

We decompose the previous problem into a separable optimization function by introducing the auxiliary vector $\mathbf{c}$ as
\begin{eqnarray}
  \min_\mathbf{c} && -\sum_i \log |c_i|^2 + {\log \|\mathbf{c}\|_2^2} \\
  s.t. && \mathbf{c = x} \\
   && \mathbf{Ax = b}
\end{eqnarray}
The optimization problem can now be written as two separate equations minimizing w.r.t to $\mathbf{x}$ and $\mathbf{c}$ as
\begin{eqnarray} \label{Eq:Min_x}
  \min_\mathbf{x} && \\
  \nonumber
  & \frac{1}{2}\|\mathbf{Ax - b}\|^2 + \eta \|\mathbf{c-x}\|^2 & \\
  \label{Eq:Min_c}
  \min_\mathbf{c} && \\
  \nonumber
  & \eta \|\mathbf{c-x}\|^2 - \sum \log |c_i|^2 + {\log \|\mathbf{c}\|_2^2}  &
\end{eqnarray}
This set of equation leads to a sparse solution of $Ax = b$, as the Lagrange multiplier $\eta \rightarrow \infty$. For a given vector $\mathbf{c}$, Eq.~\ref{Eq:Min_x}, has a unique minimizer given by
\begin{equation}\label{Eqn:Min_x_2}
  \mathbf{A}^T(\mathbf{Ax-b})-2\eta (\mathbf{c-x}) = 0
\end{equation}
which leads to the update of $\mathbf{x}$ (given fixed $\mathbf{c}$) as

\begin{equation}\label{Eqn:Min_x_3}
  \mathbf{x} = (\mathbf{A}^T\mathbf{A} + 2\eta \mathbf{I})^{-1}(\mathbf{A}^T\mathbf{b} + 2\eta \mathbf{c})
\end{equation}
A slightly faster way to solve for $\mathbf{x}$ is to compute the eigen decomposition of $\mathbf{A}^T\mathbf{A} = \mathbf{L}\mathbs{\Lambda} \mathbf{L}^T$ and then just adding $\eta$ to the eigenvalues in each iteration. This update can be written as $\mathbf{x} = \mathbf{L}(\mathbs{\Lambda} + 2\eta \mathbf{I})^{-1}\mathbf{L}^T(\mathbf{A}^T\mathbf{b} + 2\eta \mathbf{c})$.

Now let us take a closer look at Eq.~\ref{Eq:Min_c}. Let us assume that $\mathbf{c}^\star$ is a solution for Eq.~\ref{Eq:Min_c} for a given fixed $\mathbf{x}$. Then the following lemmas adopted from Duchi et al.~\cite{Duchi08} hold for the solution of Eq.~\ref{Eq:Min_c}.
\begin{lemma}
Let $p$ and $q$ be two indices such that $x_p > x_q$ . If $c^{\star}_p=0$ then $c^{\star}_q$ must be zero as well.
\end{lemma}
\begin{proof}
Assume that $c^{\star}_q$ is non zero. Lower cost can then be obtained by swapping $c^{\star}_q$ and $c^{\star}_p$ which is a contradiction since $\mathbf{c}^{\star}$ obtains the minimum cost.
\end{proof}
\begin{lemma}\label{LEMMA:2}
For all indices $i\in[1,\ldots,n]$, $x_ic^{\star}_i \ge 0$.
\end{lemma}
\begin{proof}
Assume that for a particular index $j$ the assumption does not hold and hence $x_jc^{\star}_j < 0$. Construct a new vector $\mathbf{\hat{c}}$ with all elements equal to $\mathbf{c}^{\star}$ except the element at location $j$. For this let the value for the new vector be $\hat{c}_j = -c^{\star}_j$. Since the norm terms remain unchanged in Eq.~\ref{Eq:Min_c}, the net difference in cost for the two choices can be written as 
\begin{eqnarray*}
  && \|\mathbf{x}-\mathbf{c}^{\star}\|_2^2 - \|\mathbf{x-\hat{c}}\|_2^2 \\
  &=& (x_j - c^{\star}_j)^2 - (x_j - (-c^{\star}_j))^2 \\
  &=& -2 x_j c^{\star}_j - ( 2 x_j c^{\star}_j) \\
  &=& -4 x_j c^{\star}_j > 0 ~~ (since~ x_j c^{\star}_j < 0).
\end{eqnarray*}
Hence we have constructed a solution which has lower cost than the optimal solution which is a contradiction.
\end{proof}

These two lemmas jointly enable the analysis of Eq.~\ref{Eq:Min_c} for a sorted vector in the positive orthant (Duchi et al.~\cite{Duchi08}).
Let us assume that $\mathbf{\hat{x}}$ is the vector obtained after sorting $\mathbf{{x}}$ in descending order. 
\begin{equation}\label{Eqn:x_positive}
  \mathbf{\hat{x}} = \mathrm{sort}( \mathrm{abs}(\mathbf{x}), \mathrm{descend})
\end{equation}
Also note that the sign of $c_i$ will be the same as $\hat{x}_i$ (from Lemma.~\ref{LEMMA:2}), hence we can discard the sign from the subsequent analysis and plug it back to the final result. 
In fact the cost function for minimization w.r.t $\mathbf{c}$ in Eq.~\ref{Eq:Min_c} has a striking similarity to the projection onto simplex cost function in Duchi et al.~\cite{Duchi08}. The sum over $\log |c_i|$ is similar to sum of individual elements ($l_1$-norm) in their equation. The norm value $z$ is somewhat similar to the $\log  \|\mathbf{c}\|_2^2$ term in our development. The beauty of our development is that it does not have any user supplied constants (the signal norm $z$ in \cite{Duchi08}) to deal with. Both the norm quantities are signal entities which can be estimated from the current signal.

The Lagrangian for Eq.~\ref{Eq:Min_c} 
with respect to $\mathbf{\hat{x}}$, with additional positivity constraints 
can now be written as
\begin{equation}\label{Eqn:Min_C_positive}
  L = \min_\mathbf{c}~~ \eta \|\mathbf{c-\hat{x}}\|^2 - \sum \log |c_i|^2 + {\log \|\mathbf{c}\|_2^2} - \mathbs{\zeta}. \mathbf{c}
\end{equation}
The first order optimality condition, for all positive $c_i$ can be written as
\begin{equation}\label{Eqn:Min_C_positive}
  \frac{\delta L}{\delta c_i} = 2\eta (c_i-\hat{x}_i) - \frac{2}{c_i} + \frac{2 c_i}{\|\mathbf{c}\|_2^2} - \zeta_i = 0
\end{equation}
The complementary slackness KKT condition implies that whenever $c_i > 0$ we must have that $\zeta_i = 0$. Thus, if $c_i > 0$ we get
\begin{equation}\label{Eqn:c_i}
  \eta (c_i-\hat{x}_i)c_i - 1 + \frac{c_i^2}{\|\mathbf{c}\|_2^2} = 0
\end{equation}

Summing the KKT condition in Eq.~\ref{Eqn:c_i} for all the nonzero components of $\mathbf{c}$, we get
\begin{eqnarray}
\label{Eqn:c_i_sum}
  && \sum^{\rho} \left(\eta (c_i-\hat{x}_i)c_i - 1 + \frac{c_i^2}{\|\mathbf{c}\|_2^2}\right) = 0 \\
  \label{Eqn:c_i_sum_2}
 \Rightarrow && \rho = 1 + \sum_\rho \eta (c_i-\hat{x}_i)c_i \\
 \label{Eqn:rho_2}
 \Rightarrow && \rho = 1 + \eta \mathbf{c}^T(\mathbf{c-\hat{x}})
\end{eqnarray}
where $\rho$ is the number of non-zero entries in $\mathbf{c}$.



The development till now flows seamlessly if we know the value of $\rho$, or corresponding non-zero entries in $\mathbf{c}$. We can arrive at this value by substituting Eq.~\ref{Eqn:rho_2} in Eq.~\ref{Eqn:Min_x_2}.
\begin{equation}\label{Eqn:rho_3}
  \rho = \min[1 + \mathbf{c}^T\mathbf{A}^T(\mathbf{Ax-b})/2, N]
\end{equation}
We claim that this is the first such `soft' equation to identify the support of the unknown vector. The ease of obtaining this quantity from the optimization process itself, without any apriori assumptions is one of the key novelties of our formulation.

For known value of $\rho$, 
$\mathbf{c}$ can be arrived at by finding the roots of the equation
\begin{equation}\label{Eqn:c_eqn}
\eta \mathbf{c}^T(\mathbf{c-\hat{x}})  = \rho - 1
\end{equation}

This is a family of equations which can be solved in multiple ways. The simplest method which treats all components equivalently is to complete the square.
\begin{eqnarray}
\nonumber 
  && \mathbf{c}^T\mathbf{c} -2 \mathbf{c}^T\frac{\mathbf{\hat{x}}}{2} + \frac{\mathbf{\hat{x}}^T\mathbf{\hat{x}}}{4} = \frac{\rho-1}{\eta} + \frac{\mathbf{\hat{x}}^T\mathbf{\hat{x}}}{4}\\
 \label{Eqn:c_solution}
 \Rightarrow && \|\mathbf{c}-\frac{\mathbf{\hat{x}}}{2}\|_2^2 = \frac{\rho-1}{\eta} + \frac{\|\mathbf{\hat{x}}\|_2^2}{4}
\end{eqnarray}
which is a hypersphere with its center at $\frac{\mathbf{\hat{x}}}{2}$ and radius $r$ equal to $\sqrt{\frac{\rho-1}{\eta} + \frac{\mathbf{\hat{x}}^T\mathbf{\hat{x}}}{4}}$. 
%

Another possible solution for the system of equations in Eq.~\ref{Eqn:c_eqn} can be obtained by noting that it can be written as
\begin{equation}\label{Eqn:c_eqn_2}
c_1(c_1-\hat{x}_1) + \sum_{i=2}^{\rho}( c_i(c_i-\hat{x}_i) - 1/\eta)   = 0
\end{equation}
where the largest component $c_1$ of $\mathbf{c}$ is exactly matched with its counterpart $\hat{x}_1$, while the other smaller components are scaled. Empirically this solution seems to work better than the uniform scaling technique and accounts for the special case of $\rho=1$. Finally, solving the quadratic equation and choosing the positive root, the update for $\mathbf{c}$ can be written as
\begin{equation}\label{Eqn:c_update_scaled}
c_i = \hat{x}_i \frac{1 + \sqrt{1 + \frac{4}{\eta x_i^2}[i \neq 1]}}{2}
\end{equation}
where $[.]$ is the indicator function. 


Finally, reshuffling the entries to match the original vector $\mathbf{x}$ and putting the sign of $\mathbf{x}$, we get
\begin{equation}\label{Eqn:final_C}
  \mathbf{c} \Leftarrow \mathbf{c}[ \pi(\mathbf{\hat{x}} \rightarrow \mathbf{x})]* \mathrm{sign}( \mathbf{x})
\end{equation}
where $\pi(\mathbf{\hat{x}} \rightarrow \mathbf{x})$ is the inverse mapping from the sorted vector $\mathbf{\hat{x}}$ to the original vector $\mathbf{x}$ and $*$ denotes elementwise multiplication. Note that the scaling step and the rearrangement step can be handled together by knowing the permutation $\pi$ only, witout obtaining the sorted vector $\mathbf{\hat{x}}$.

\begin{algorithm}                      
\caption{Solve-Select-Scale Algorithm}          
\label{alg1}                           
\begin{algorithmic}                    
    \REQUIRE $\mathbf{A,b}$, $\eta_{start}$, $\eta_{end}$, $\epsilon$ 
    \STATE	$\eta \leftarrow \eta_{start}$
    \STATE	$\mathbf{c} \leftarrow \mathbf{1}(size(\mathbf{x}))$
    \WHILE{$\eta < \eta_{end}$}
        \STATE SOLVE $\mathbf{x}$~Eq.~\ref{Eqn:Min_x_3}
        \STATE SELECT $\rho$~Eq.~\ref{Eqn:rho_3}
        \STATE SCALE  $\mathbf{c}[\rho]$~Eq.~\ref{Eqn:c_update_scaled}, set all other entries to 0
	 \STATE   $\eta \leftarrow (1+\epsilon)\eta$
    \ENDWHILE
\end{algorithmic}
\end{algorithm}


\begin{figure}[ht]
\begin{center}
\includegraphics[width=8cm,height=3cm]{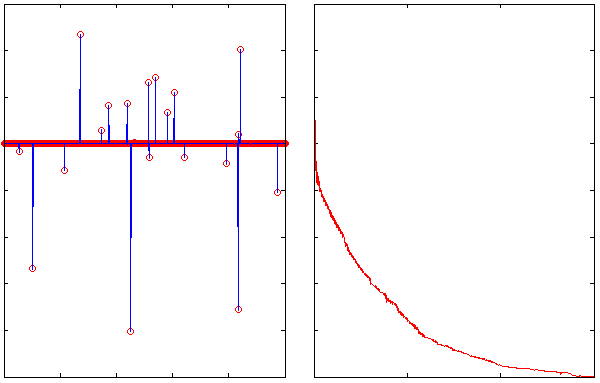}
\includegraphics[width=8cm,height=3cm]{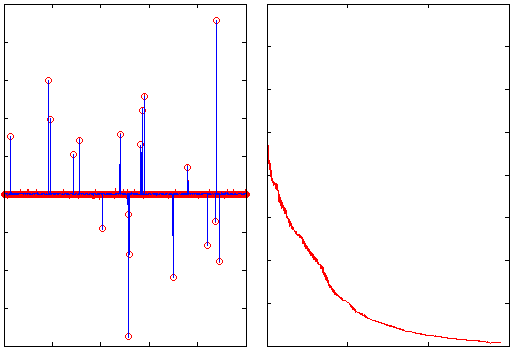}
\caption{Simulation results. Left: Original signal (red circles) and its reconstruction in  blue $\{N,K,M\} = \{1000,20,600\}$. Right: error norm $\|\mathbf{Ax-b}\|^2$ for 1500 iterations.}
\label{Fig:ReconstructionRes}
\end{center}
\end{figure}

\section{Convergence}
We borrow heavily from the convergence analysis for additive half quadratic method as developed my Nikolova and Ng.~\cite{Nikolova05}. Let us write the unconstrained cost function as
\begin{equation}\label{Eqn:conv_cost}
  J(\mathbf{x}):= f(\mathbf{x}) +  \beta (\mathbf{Ax - b})^2
\end{equation}
where $f(.)$ is the sparsity producing metric given in Eq.~\ref{EQN:SparseMetric}. Now we can introduce the auxillary variable $\mathbf{c}$ as
\begin{equation}\label{Eqn:conv_cost}
  \mathcal{J}(\mathbf{x},\mathbf{c}):= f(\mathbf{x}) +  \beta \left[\eta(\mathbf{c-x})^2 + (\mathbf{Ac - b})^2\right]
\end{equation}
Note that for $\eta \rightarrow \infty$
\begin{equation}\label{Eqn:conv_cost}
  J(\mathbf{x}) \approx \min_{\mathbf{c}}\mathcal{J}(\mathbf{x},\mathbf{c})
\end{equation}
For a fixed $\eta$, at iteration $k+1$ we evaluate
\begin{eqnarray}
\nonumber 
 \mathbf{c}^{(k+1)} &\leftarrow& \mathcal{J}(\mathbf{x}^{(k)},\mathbf{c}^{(k+1)}) \le \mathcal{J}(\mathbf{x}^{(k)},\mathbf{c})\\
\nonumber 
 \mathbf{x}^{(k+1)} &\leftarrow& \mathcal{J}(\mathbf{x}^{(k+1)},\mathbf{c}^{(k+1)}) \le \mathcal{J}(\mathbf{x},\mathbf{c}^{(k+1)})
\end{eqnarray}
It can be easily verified by substituting $\{\mathbf{x}^{(k)},\mathbf{c}^{(k)}\}$ for the index free $\{\mathbf{x},\mathbf{c}\}$ above that the iterates form a decreasing sequence
\begin{equation}\nonumber
  \{\ldots \mathcal{J}(\mathbf{x}^{(k)},\mathbf{c}^{(k)}) \ge \mathcal{J}(\mathbf{x}^{(k+1)},\mathbf{c}^{(k+1)})\ge \ldots\}
\end{equation}
which according to Eq.~\ref{Eqn:conv_cost} 
leads to 
\begin{equation}\nonumber
  \{\ldots {J}(\mathbf{x}^{(k)}) \ge {J}(\mathbf{x}^{(k+1)})\ge \ldots\}
\end{equation}
Also note that the constant $\beta=1$ can be absorbed into the formulation resulting in the development mentioned in Sec.~\ref{SEC:Optimization}. 

\section{Comparison with other optimization schemes}
We would like to bring out the differences between our method and traditional optimization schemes. We would like to compare against Frank-Wolfe~\cite{FrankWolfe1956} method which proceeds by solving linear problems around the current iterate. At this point we point that the gradient as well as the Hessian of the cost function (Eq.~\ref{EQN:SparseMetric}) become unbounded whenever some elements of the unknown vector $\mathbf{x}$ go to zero. Consequently, these entries have to be tracked independently to make sure that the optimization remains tractable. However, once some entry has been removed from the optimization, it remains uncertain whether (or when) it should be reintroduced (or not). Hence, any random initialization, with some known entries very close to zero, will never be able to put non-zero values to these entries of the unknown vector $\mathbf{x}$. For our method this part is already taken care of and we do not need to remove any elements (or subsequently add them later) based on their values in the current iterate. The main benefit of our scheme is the fact that the method identifies the support as well as minimizes the cost function. The number of non-zero entries is determined by the term $\rho$ and the unique sorting of the indices takes care of the selection criteria.

A closer look at the cost function in Eq.~\ref{EQN:SparseMetric} reveals that it is a combination of one convex function $-\sum \log |x_i|^2$ and a concave term $\log \mathbf{x^Tx}$. As detailed in the concave-convex procedure by Yuille and Rangarajan~\cite{Yuille2003}, the cost function indeed steps through a saw-tooth type profile while being consistently decreased in the experiments. This can be observed in the zoomed up profile of the cost function as shown in Fig.~\ref{Fig:sawtooth}.
\begin{figure}[ht]
\begin{center}
\includegraphics[width=8cm]{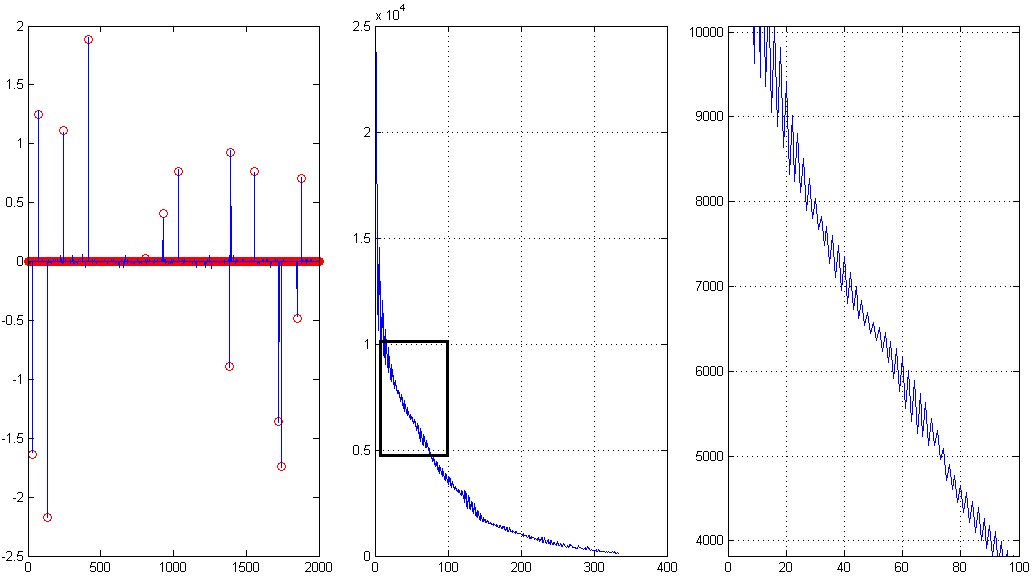} 
\caption{Simulation results for noiseless case. $N=1000$, $K=15$, $M = \{500\}$. The region inside the box in the centre panel is zoomed in the right panel. 
}
\label{Fig:sawtooth}
\end{center}
\end{figure}  

%

\section{Experiments}
First set of experiments work with noiseless simulations to evaluate the system introduced in this paper. We set the data dimension $n=1000$, the sparsity $k=\{5,10,15,20,25\}$ and number of measurements $m=\{600,700,800,900\}$. We perform 50 rounds of simulation for each combination of $\{n,k,m\}$ and run each simulation set for 1500 iterations. Fig.~\ref{Fig:ReconstructionRes} shows two runs for the simulations for $\{N,K,M\} = \{1000,20,600\}$. All the `on' entries of the original vector are identified for the noiseless case. Next, we study the ratio $\hat{s}/s_0$ for different number of measurements $m$ where $\hat{s}$ is the sparsity metric evaluated for the reconstucted signal and $s_0$ is the robust sparsity for the original signal. 

\begin{figure}[ht]
\begin{center}
\includegraphics[width=8cm]{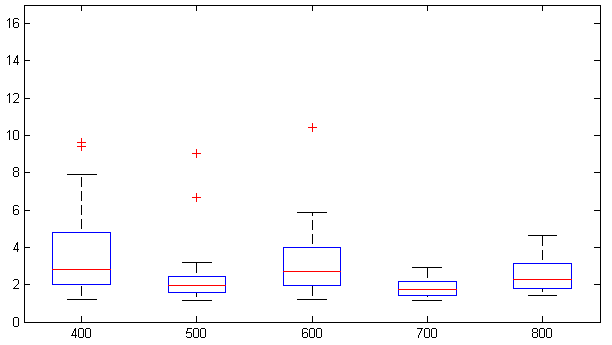}
\caption{Simulation results for noiseless case. $N=1000$, $K=20$, $M = \{500,600,\ldots,900\}$. Note the median value for the ratio $\hat{s}/s$ ranges between 1.5 to 3.}
\label{Fig:sEstimateNoiseless}
\end{center}
\end{figure} 

For the next set of experiments, we evaluate the performance of our algorithm under noisy conditions. We add noise to the measurements and try to solve a slightly modified problem
\begin{eqnarray}
  \min_\mathbf{x} && -\sum_i \log |x_i|^2 + {\log \|\mathbf{x}\|_2^2} \\
\label{Eq.StoppingCondNoise}
  s.t. &&  (\mathbf{Ax - b})^2 < \sigma^2
\end{eqnarray}
where $\sigma^2$ is a function of the noise variance. From an algorithm perspective, we modify the stopping criteria as Eq.~\ref{Eq.StoppingCondNoise}. The results for varying $m$ are shown in Fig.~\ref{Fig:sEstimateNoisy}.
\begin{figure}[ht]
\begin{center}
\includegraphics[width=8cm]{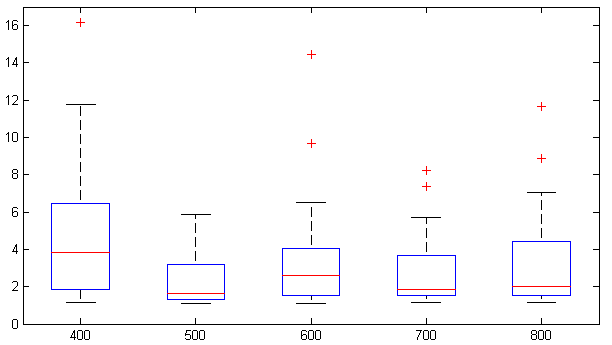}
\caption{Simulation results for noisy case. $N=1000$, $K=20$, $M = \{300,400,\ldots,700\}$, $\sigma^2 = .01$.}
\label{Fig:sEstimateNoisy}
\end{center}
\end{figure} 

\subsection{Comparison to COSAMP}
Our method is very close to the compressive sampling matching pursuit (COSAMP) proposed by Tropp and Needell~\cite{Tropp08}. Concentrating on random Gaussian matrices as our measurement system, we can make sure that all the requirements of the restricted isometric property (Candes and Tao~\cite{Candes06}) are taken care of. The sorting step is the same as the identification phase, least square estimation is the same as the solve for $\mathbf{x}$ stage. The key difference is the fact that for COSAMP to converge a good enough estimation of the unknown support is needed. Once this requirement is not satisfied, COSAMP starts to loose the sparsity requirement with alarming rate.
\begin{figure*}[ht]
\begin{center}
\includegraphics[width=16cm,height=7cm]{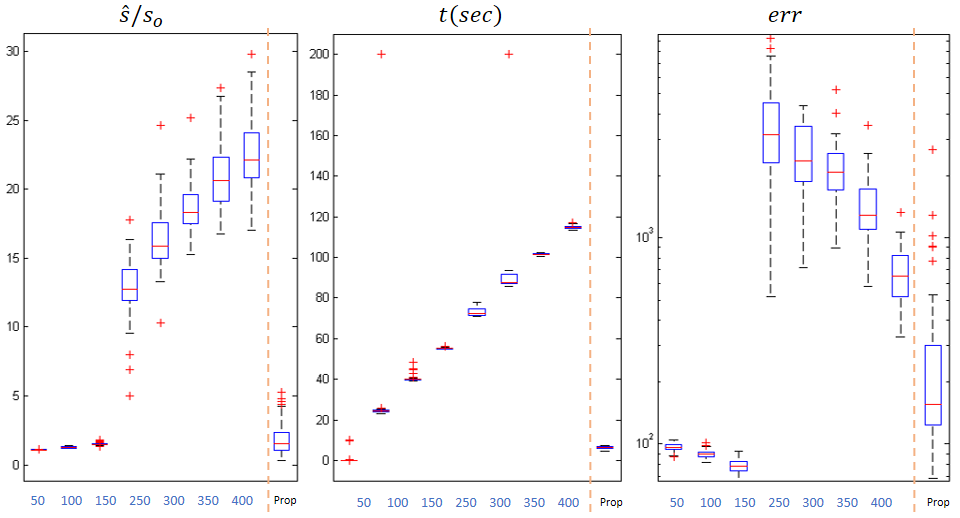}
\caption{Comparison against COSAMP. $N=1000$, $K=20$, $M = 500$. Left to right: sparsity estimate $\hat{s}/s_o$, runtime $t$ in seconds and reconstruction error $\|\mathbf{Ax-b}\|$ for 50 random runs with each setting. The sparsity estimate `k' provided to COSAMP is shown at the horizontal axis. For all the charts, last plot to the right is for the proposed method (prop). Note that the sparsity estimate is more than 10 consistently for COSAMP once `k' moves far from the true estimate. All the experiments were performed with noisy data.}
\label{Fig:CompCosamp}
\end{center}
\end{figure*} 

\section{Conclusion}
In this paper we introduce a new algorithm to reconstruct sparse signals from noisy measurements of the form $\mathbf{Ax=b+\epsilon}$. The proposed method does not need any sparsity estimate for the unknown signal. The solve-select-scale paradigm estimates all the required parameters from the signal itself and can be implemented as an iterative routine, with only a least square solver. The ease of formulation and signal estimation leads to a very simple algorithm. We report noiseless experiments to establish the merits of the algorithm. Finally, we provide extensive comparison against COSAMP (Tropp and Needell~\cite{Tropp08}) and show how the performance of COSAMP is heavily tied to the initial sparsity estimate of the unknown signal provided as a parameter to the algorithm. Our proposed method, on the other hand, is not dependent on any such parameter. In the future, we would like to continue to experiment with different class of measurement matrices.

\bibliography{example_paper}
\bibliographystyle{acl2016} 

\end{document}